\documentclass[twocolumn,pra,showpacs,superscriptaddress,notitlepage,amssymb,preprintnumbers]{revtex4-1}

\usepackage{amsmath,amssymb,amsthm}
\usepackage{braket}
\usepackage{empheq}
\usepackage{enumitem}
\usepackage{subfigure}
\usepackage{comment}

\usepackage[pdfstartview=FitH]{hyperref}
\hypersetup{
    colorlinks=true,       
    linkcolor=cyan,          
    citecolor=magenta,        
    filecolor=magenta,      
    urlcolor=cyan,           
    runcolor=cyanf
}

\newcommand{\bes} {\begin{subequations}}
\newcommand{\ees} {\end{subequations}}
\newcommand{\beq}{\begin{equation}}
\newcommand{\eeq}{\end{equation}}
\newcommand{\bea}{\begin{eqnarray}}
\newcommand{\eea}{\end{eqnarray}}

\newcommand\Tr{\mathrm{Tr}}

\newtheorem{theorem}{Theorem}
\newtheorem*{theorem*}{Theorem}

\newtheorem{lemma}{Lemma}
\newtheorem{definition}{Definition}

\begin{document}
\title{On the Computational Complexity of Curing the Sign Problem}

\author{Milad Marvian}
\affiliation{Department of Electrical Engineering, University of Southern California, Los Angeles, California 90089, USA}
\affiliation{Center for Quantum Information Science \&
Technology, University of Southern California, Los Angeles, California 90089, USA}

\author{Daniel A. Lidar}
\affiliation{Department of Electrical Engineering, University of Southern California, Los Angeles, California 90089, USA}
\affiliation{Center for Quantum Information Science \&
Technology, University of Southern California, Los Angeles, California 90089, USA}
\affiliation{Department of Physics and Astronomy, University of Southern California, Los Angeles, California 90089, USA}
\affiliation{Department of Chemistry, University of Southern California, Los Angeles, California 90089, USA}
\author{Itay Hen}
\affiliation{Center for Quantum Information Science \& Technology, University of Southern California, Los Angeles, California 90089, USA}
\affiliation{Department of Physics and Astronomy, University of Southern California, Los Angeles, California 90089, USA}
\affiliation{Information Sciences Institute, University of Southern California, Marina del Rey, CA 90292}

\begin{abstract}
Quantum many-body systems whose Hamiltonians are non-stoquastic, i.e., have positive off-diagonal matrix elements in a given basis, are known to pose severe limitations on the efficiency of Quantum Monte Carlo algorithms designed to simulate them, due to the infamous sign problem. 
We study
the computational complexity associated with `curing' non-stoquastic Hamiltonians, i.e., transforming them into sign-problem-free
ones. We prove that if such transformations are limited to single-qubit Clifford group elements or general single-qubit orthogonal matrices, finding the curing transformation is NP-complete. We discuss the implications of this result.
\end{abstract}
\maketitle

\textit{Introduction.}---The `negative sign problem', or simply the `sign problem'~\cite{Loh-PRB-90}, is the single most important unresolved challenge in quantum many-body simulations, preventing physicists, chemists, and material scientists alike from 
being able to efficiently simulate
many of the most profound macroscopic quantum physical phenomena of Nature, in areas as diverse as high temperature superconductivity and material design through neutron stars to lattice quantum chromodynamics.
More specifically, the sign problem slows down Quantum Monte Carlo (QMC) algorithms~\cite{Landau:2005:GMC:1051461,newman}, which are in many cases the only practical method available for studying large quantum many-body systems, 
to the point where they become practically useless. QMC algorithms evaluate thermal averages of physical observables by the (importance-) sampling of quantum configuration space via the decomposition of the partition function into a sum of easily computable terms, or weights, which are in turn interpreted as probabilities in a Markovian process. Whenever this decomposition contains negative terms, QMC methods tend to converge exponentially slowly.  Most dishearteningly, it is typically the systems with the richest quantum mechanical behavior that exhibit the most severe sign problem. 

In defining the scope under which QMC methods are 
sign-problem free, 
the concept of `stoquasticity', first introduced by Bravyi et al.~\cite{Bravyi:QIC08}, has recently become central.~%
The most widely used definition of a local stoquastic Hamiltonian is:  
\begin{definition}[\cite{Bravyi:2009sp}]
\label{def:1}
A local Hamiltonian, $H =  \sum_{a=1}^M H_a$
is called stoquastic with respect to a basis ${\cal B}$ iff all $H_a$ have only non-positive off-diagonal matrix elements in the basis ${\cal B}$.
\end{definition}
In the basis ${\cal B}$, the partition function decomposition of stoquastic Hamiltonians leads to a sum of strictly non-negative weights and such Hamiltonians hence do not suffer from the sign problem~\footnote{E.g., 
in the path-integral formulation of QMC with respect to a basis ${\cal{B}}=\{ b \}$, the partition function 
$Z$ is reduced to an $L$-fold product of sums over complete sets of basis states, $\{b_1\},\ldots,\{b_L\}$, which are weighted by the size of the imaginary-time slice $\Delta \tau=\beta/L$ and the off-diagonal matrix elements of $H$. Namely,
$Z \approx \prod_{l=1}^{L} \sum_{b_l} \bra{b_l}  \mathrm{e}^{-\Delta \tau H_{l,l+1}}\ket{b_{l+1}}$,
where  $L$ is the number of slices and periodic boundary conditions are assumed. 
The connection to stoquasticity is that when all the off-diagonal matrix elements, $H_{j,j+1}$ in the given basis are non-positive, these weights are purely positive for each time slice.
}.
On the other hand, non-stoquastic Hamiltonians, whose local terms have positive off-diagonal entries induce negative weights and generally lead to the sign problem~\cite{Loh-PRB-90,Troyer2005} unless certain symmetries are present. 

The concept of stoquasticity is also important from a computational complexity theory viewpoint. For example, the complexity class StoqMA associated with the problem of deciding whether the ground state energy of stoquastic local Hamiltonians is above or below certain values, is expected to be strictly contained in the complexity class QMA, that poses the same decision problem for general local Hamiltonians~\cite{Bravyi:QIC08}. Additionally, StoqMA appears as an essential part of the dichotomy of two-local qubit Hamiltonians problem~\cite{Cubitt:2016vl}. 

However, stoquasticity does not imply efficient equilibration%
~\footnote{Throughout this work we reserve the term `efficient' to mean that the algorithm requires at most a polynomial run-time in the problem size (the number of variables). For example, an algorithm  equilibrates efficiently if it can correctly samples from the Gibbs distribution of the Hamiltonian in question  with at most a polynomial run-time in the problem size  and the given statistical error [see part (iii) of the definition of the sign problem given in Ref.~\cite{Troyer2005} for a precise statement.]}. 
E.g., finding the ground state energy of a classical Ising model---which is trivially stoquastic---%
is already NP-hard~\cite{Barahona1982}. 
Conversely, non-stoquasticity does not imply inefficiency: 
there exist numerous cases where an apparent sign problem (i.e., non-stoquasticity) is the result of a naive basis choice that can be transformed away, resulting in efficient equilibration~\cite{PhysRevB.89.134422,Alet,PhysRevB.93.054408}. 

Here we focus on the latter, i.e, whether non-stoquasticity can be `cured'. To this end we propose an alternative definition of stoquasticity that is based on the computational complexity associated with transforming non-stoquastic Hamiltonians into stoquastic ones. 

\textit{Stoquasticity revisited}.---To motivate our alternative definition, we first note 
that any Hamiltonian can trivially be presented as stoquastic via diagonalization. However, the complexity of finding the diagonalizing basis generally grows exponentially with the size of the system (as noted in Ref.~\cite{Troyer2005}) and the new basis will generally be highly non-local and hence not efficiently representable. 
We also note that it is 
straightforward to construct examples where apparent non-stoquasticity may be transformed away. For example, consider the $n$-spin Hamiltonian $H_{XZ}= \sum \tilde{J}_{ij} X_iX_j- \sum J_{ij} Z_i Z_j$ with non-negative $J_{ij},\tilde{J}_{ij}$, where $X_i$ and $Z_i$ are the Pauli matrices acting on spin $i$. This Hamiltonian is non-stoquastic, but can easily be converted into stoquastic form.
Denoting the Hadamard gate (which swaps $X$ and $Z$) by $W$,
consider the transformed Hamiltonian 
$H_{ZX}=W^{\otimes n}  H_{XZ} W^{\otimes n}= -\sum J_{ij} X_iX_j+ \sum \tilde{J}_{ij}  Z_i Z_j$ which is stoquastic. The sign problem of the original Hamiltonian, $H_{XZ}$,  can thus be efficiently cured by 
a unit-depth circuit of single-qubit rotations.
Moreover, thermal averages are invariant under 
unitary transformations \footnote{{The thermal average is $\protect\langle A \protect\rangle \protect\equiv \protect\frac{1}{Z}\Tr(e^{-\beta H} A)$ where $\protect\beta$ is the inverse temperature, and our claim is that $\protect\langle A \protect\rangle = \protect\langle A' \protect\rangle$ where $A' = UAU^\dagger$ for unitary $U$. Here is the proof. First, $Z=\Tr(e^{-\beta H})=\Tr( U e^{ -\beta H } U^\dagger) = \Tr( e^{ -\beta H'} ) = Z'$ where $H' = UHU^\dagger$. Thus $\protect\langle A \protect\rangle= \frac{1}{Z}\Tr(U e^{-\beta H} U^\dagger UA U^\dagger) = \frac{1}{Z'}\Tr(e^{-\beta H'} A') = \protect\langle A' \protect\rangle$.}},
so that if QMC is run on the transformed, stoquastic Hamiltonian, it is no longer slowed down by the sign problem. 
Finally, note that Definition~\ref{def:1} implies that a local Hamiltonian, $H =  \sum_{a=1}^M H_a$, is stoquastic if \textit{all} terms $H_a$ are stoquastic. However, there always remains some arbitrariness in the manner in which the total Hamiltonian is decomposed into the various terms.  Consider, e.g., $H=-2X_1 +X_1Z_2$. The second term  separately is non-stoquastic, whereas the sum is stoquastic. This suggests that the grouping of terms matters [see the Supplemental Information (SI) Sec.~\ref{SI1} for a method to find such a grouping of terms].

The above considerations motivate a reexamination of the concept of stoquasticity from 
a {complexity-theory} perspective, which can have important consequences for QMC simulations~\footnote{A related approach was discussed by Barbara Terhal at the AQC17 conference~\url{http://www.smapip.is.tohoku.ac.jp/~aqc2017/program.html}.}.
For example, given a $k$-local non-stoquastic Hamiltonians $H = \sum_a H_a$ (where each summand is a $k$-local term, i.e., a tensor product of at most $k$ non-identity single-qubit Pauli operators), we may ask whether there exists a constant-depth quantum circuit $U$ such that 
$H'=U H U^{\dagger}$ can be written as a $k'$-local stoquastic Hamiltonian $H'=\sum_a H'_a$ and if so, what the complexity associated with finding it is. 
It is the answer to the latter question that determines whether the Hamiltonian in question should be considered \emph{computationally} stoquastic, i.e., whether it is feasible (in a complexity theoretic sense) to find a `curing' transformation $U$, which would then allow QMC to compute thermal averages with $H$ by replacing it with $H'$. More formally, we propose the following definition:

\begin{definition}
A unitary transformation $U$ `cures' a non-stoquastic Hamiltonian $H$ (i.e., removes its sign problem) represented in a given basis if $H'=UHU^\dagger$ is stoquastic, i.e., its off-diagonal elements in the given basis are all non-positive. A family of local Hamiltonians $\{H\}$ represented in a given basis is \emph{efficiently curable} (or, equivalently, \emph{computationally stoquastic}) if there exists a polynomial-time classical algorithm such that for any member of the family $H$, the algorithm can find a unitary $U$ with the property that $H'=UHU^\dagger$ is local and stoquastic in the given basis. 
\end{definition}
As an example, the Hamiltonian $H_{XZ}$ considered above is efficiently curable. General local Hamiltonians are unlikely to be efficiently curable as this would imply the implausible result that QMA$=$StoqMA~\cite{hastings2016quantum}.

Note that given some class of basis transformations, our definition distinguishes between the ability to cure a Hamiltonian \emph{efficiently} or in \emph{in principle}. For example, deciding whether a Pauli group element $U=\bigotimes_{i=1}^n u_i$, where $u_i$ belongs to the single-qubit Pauli group $\mathcal{P}_1=\{I,X,Y,Z\}\times\{\pm 1,\pm i\}$, can cure each term $\{H_a\}$ of a $k$-local Hamiltonians $H=\sum_a H_a$, can be solved in polynomial time (see the SI, Sec.~\ref{SI:Pauli}). However, the Hamiltonian $H=X_1 Z_2$ cannot be made stoquastic in principle using a Pauli group element, as conjugating it with Pauli operators results in $\pm X_1 Z_2$, both of which are non-stoquastic. (See Ref.~\cite{hastings2016quantum} for results on an intrinsic sign problem for local Hamiltonians.)

Going beyond Pauli operators, 
our main result is a proof that even for particularly simple local transformations such as the single-qubit Clifford group and real-valued rotations the problem of deciding whether a family of local Hamiltonians is curable cannot be solved efficiently, in the sense that it is equivalent to solving 3SAT and is hence NP-complete.

We assume that a $k$-local Hamiltonian $H=\sum_a H_a$ is described by specifying each of the local terms $H_a$, and the goal is to find a unitary $U$ that cures each of these local terms. In general, a unitary $U$ that cures the total Hamiltonian $H$ may not necessarily cure all $H_a$ separately. However, for all of the constructions in this paper we prove that a unitary $U$  cures $H$ if and only if it cures all  $H_a$ separately. The decomposition $\{H_a\}$ is merely used to guarantee that verification is efficient and the problem is contained in NP. 

\textit{Complexity of curing the sign problem for the single-qubit Clifford group}.---%
To study the computational complexity associated with finding a curing transformation $U$, 
we shall consider for simplicity single-qubit unitaries $U=\bigotimes_{i=1}^n u_i$ and only real-valued Hamiltonian matrices. As we shall show, even subject to these simplifying restrictions the problem of finding a curing transformation $U$ is computationally hard.

We begin by considering the computational complexity of finding local rotations from a discrete and restricted set of rotations. Specifically, we consider the single-qubit Clifford group $\mathcal{C}_1$  (with group action defined as conjugation by one of its elements), defined as $\mathcal{C}_1 = \{U\ | \ UgU^\dagger \in \mathcal{P}_1\ \forall g\in \mathcal{P}_1\}$, i.e., the normalizer of $\mathcal{P}_1$. It is well known that $\mathcal{C}_1$ is generated by $W$ and the phase gate $P=\mathrm{diag}(1,i)$ \cite{nielsen2010quantum}.
\begin{theorem} \label{thm1}
Let $U=\bigotimes_{i=1}^n u_i$, where $u_i$ belongs to the single-qubit Clifford group. Deciding whether there exists a curing unitary $U$ for $3$-local Hamiltonians is NP-complete.
\end{theorem}
We prove this theorem by reducing the problem 
to the canonical NP-complete problem known as 3SAT ($3$-satisfiability)~\cite{Karp:21-problems}, beginning with the following lemma:
\begin{lemma} \label{lemmaIW}
	Let $u_i\in \{I,W_i\}$, where $I$ is the identity operation and $W_i$ is a Hadamard gate on the $i$-th qubit. Deciding whether there exists a curing unitary $U=\bigotimes_{i=1}^n u_i$ for $3$-local Hamiltonians is NP-complete.
\end{lemma}
To prove Lemma~\ref{lemmaIW}, we first introduce a mapping between 3SAT and $3$-local Hamiltonians. Our goal is to find an assignment of $n$ binary variables $x_i\in\{0,1\}$ such that the unitary $W(x)\equiv\bigotimes_{i=1}^{n} W_i^{x_i}$ [where $x\equiv (x_1,\dots,x_n)$] rotates an input Hamiltonian to a stoquastic Hamiltonian.
We use the following $3$-local Hamiltonian as our building block:
\bea
H_{ijk}^{(111)} =  Z_i Z_j Z_k &-& 3 (Z_i +Z_j +Z_k)\nonumber \\
&-&(Z_i Z_j+ Z_i Z_k+Z_j Z_k) \ ,
\label{eq:Hijk111}
\eea
where $i,j$ and $k$ are three different qubit indices.
It is straightforward to check that 
\bea
\label{eq:2}
W(x) H_{ijk}^{(111)}W^{\dagger}(x)&=&  \\
W_i^{x_i} \otimes W_j^{x_j} &\otimes& W_k^{x_k}( H_{ijk}^{(111)} )W_i^{x_i} \otimes W_j^{x_j} \otimes W_k^{x_k} \nonumber
\eea
is stoquastic (``True") for any combination of the binary variables $(x_i,x_j,x_k)$ except for $(1,1,1)$, which makes Eq.~\eqref{eq:2} non-stoquastic (``False").
This is precisely the truth table for the 3SAT clause $(\bar{x}_i \vee \bar{x}_j \vee \bar{x}_k)$,
 where $\vee$ denotes the logical disjunction and bar denotes negation.
We can define the other seven possible 3SAT clauses by conjugating $H_{ijk}^{(111)} $ with Hadamard or identity gates:
\beq
\label{eq:Hijkabc}
H_{ijk}^{(\alpha\beta\gamma)} =   
W_i^{\bar{\alpha}} \otimes W_j^{\bar{\beta}} \otimes W_k^{\bar{\gamma}} \cdot H_{ijk}^{(111)} \cdot W_i^{\bar{\alpha}} \otimes W_j^{\bar{\beta}} \otimes W_k^{\bar{\gamma}}\ .
\eeq
The Hamiltonian $W(x) H_{ijk}^{(\alpha\beta\gamma)}W^{\dagger}(x)$ is non-stoquastic (corresponds to a clause that evaluates to False) only when $(x_i,x_j,x_k)=(\alpha,\beta,\gamma)$, and is stoquastic (True) for any other choice of the variables $x$.
We have thus established a bijection between $3$-local Hamiltonians $H_{ijk}^{(\alpha\beta\gamma)}$, with $(\alpha,\beta,\gamma)\in\{0,1\}^3$, and the eight possible 3SAT clauses on three variables $(x_i,x_j,x_k)\in\{0,1\}^3$. We denote these clauses, which evaluate to False iff $(x_i,x_j,x_k) = (\alpha,\beta,\gamma)$, by $C_{ijk}^{(\alpha\beta\gamma)}$.

The final step of the construction is to add together such `3SAT-clause Hamiltonians' to form 
\beq
\label{eq:H3SAT}
H_{\mathrm{3SAT}}= \sum_{C} H_{ijk}^{(\alpha\beta\gamma)} \,,
\eeq
where $C$ is the set of all $M$ clauses in the given 3SAT instance $\bigwedge C_{ijk}^{(\alpha\beta\gamma)}$. 
Having established a bijection between 3SAT clauses and 3SAT-clause Hamiltonians, the final step is to show that curing every member of the family of $H_{\mathrm{3SAT}}$ Hamiltonians by applying Hadamards as necessary, i.e., finding $x$ such that  
\beq
H'=W(x) H_{\mathrm{3SAT}}W(x)
\eeq
is stoquastic for every $H_{\mathrm{3SAT}}$, is equivalent to solving the NP-complete problem of finding satisfying assignments $x$ for the corresponding 3SAT instances. To prove the equivalence we show (i) that satisfying a 3SAT instance implies that the corresponding $H_{\mathrm{3SAT}}$ is cured, and (ii) that if $H_{\mathrm{3SAT}}$ is cured this implies that the corresponding 3SAT instance is satisfied.

(i) Note that any assignment $x$ that satisfies the given 3SAT instance also satisfies each individual clause. It follows from the bijection we have established that such an assignment cures each corresponding 3SAT-clause Hamiltonian individually. 
The stoquasticity of $H'$ then follows by noting that the tensor product of a stoquastic Hamiltonian with the identity matrix is still stoquastic and the sum of stoquastic Hamiltonians is stoquastic. 

(ii) We prove that an unsatisfied 3SAT instance implies that the corresponding $H_{\mathrm{3SAT}}$ is not cured. 
It suffices to focus on a particular clause $C_{ijk}^{(\alpha\beta\gamma)}$. The choice of variables that makes this clause False rotates the corresponding 3SAT-clause Hamiltonian to one that contains a non-stoquastic $+X_i X_j X_j$ term, which generates positive off-diagonal elements in specific locations in the matrix representation of $H_{\mathrm{3SAT}}$. Since no other 3SAT-clause Hamiltonian in $H_{\mathrm{3SAT}}$ contains an identical $X_i X_j X_k$ term, these positive off-diagonal elements cannot be cancelled out or made negative regardless of the choice of the other variables in the assignment. The rotated Hamiltonian is therefore guaranteed to be non-stoquastic. 

This establishes that the problem is NP-Hard. Checking whether a given $U$ cures all the local terms $\{H_a \}$ is efficient and therefore the problem is NP-complete.  $\qedsymbol$

To complete the proof of Theorem~\ref{thm1}, let us consider the modified Hamiltonian
\beq
\label{eq:tildeH}
\tilde{H}_{\mathrm{3SAT}}=H_{\mathrm{3SAT}} + c H_0\ , \quad H_0= - \sum_{i=1}^{n}(X_i + Z_i)\ ,
\eeq
where $c=O(1)M$ and $H_0$ is manifestly stoquastic%
~\footnote{In general, $M$ can grow polynomially (usually linearly for hard instances) with $n$, but $\tilde{H}_{\mathrm{3SAT}}$ remain $k$-local since the operator norm of each of its terms is $\textrm{poly}(n)$, as required by the definition of a $k$-local Hamiltonian (see, e.g., Definition~1 of Ref.~\cite{Cubitt:2016vl}).}. 
The goal is to find a unitary $U=\bigotimes_{i=1}^n u_i$ with $u_i \in \mathcal{C}_1$ that cures $\tilde{H}_{\mathrm{3SAT}}$. Note first that any choice of $U=\otimes_{i=1}^{n} W_i^{x_i}$ that cures $H_{\mathrm{3SAT}}$ also cures $\tilde{H}_{\mathrm{3SAT}}$. Second, note that choosing any $u_i$ that keeps $H_0$ stoquastic is equivalent to choosing one of the elements of $\mathcal{C}'_1\equiv \{I,X,W,XW\}\subset\mathcal{C}_1$ (e.g., the phase gate, which is an element of $\mathcal{C}_1$, maps $X$ to $Y$ so is excluded, as is $WX$, which maps $Z$ to $-X$). Therefore, by choosing $c$ to be large enough,
 any choice of $u_i\in \mathcal{C}_1\setminus\mathcal{C}'_1$ would transform $\tilde{H}_{\mathrm{3SAT}}$ into a non-stoquastic Hamiltonian. It follows that if $u_i \in \mathcal{C}_1$ and is to cure $\tilde{H}_{\mathrm{3SAT}}$ then in fact it must be an element of $\mathcal{C}'_1$.

Next, we note that conjugating a matrix by a tensor product of $X$ or identity operators only shuffles the off-diagonal elements but never changes their values (for a proof see the SI, Sec.~\ref{SI3}).
Therefore, for the purpose of curing a Hamiltonian, applying $X$ is equivalent to applying $I$ and applying $XW$ is equivalent to applying $W$. With this observation, the set of operators that can cure a Hamiltonian is effectively reduced from $\{I,W,XW,X\}$ to $\{I,W\}$.  According to Lemma~\ref{lemmaIW}, deciding whether such a curing transformation exists is NP-complete. $\qedsymbol$

\textit{Complexity of curing the sign problem for the single-qubit orthogonal group}.---%
We now consider the single-qubit orthogonal group, i.e., transformations of the form $Q=\bigotimes_{i=1}^n q_i$, where $q_i^T q_i = I$ $\forall i$. We first note that there exist Hamiltonians which cannot be cured in principle using any $Q$ (see the SI, Sec.~\ref{SI2}). Next, 
we show that the problem of curing the sign problem remains NP-complete when the set of allowed rotations is extended to arbitrary single-qubit orthogonal matrices. Namely:
\begin{theorem} 
\label{thm2}
Deciding whether there exists a curing orthogonal transformation $Q$ for $6$-local Hamiltonians is NP-complete.
\end{theorem}
\begin{proof}
We relegate some details to the SI (Sec.~\ref{SI4}). The proof builds on that of Theorem~\ref{thm1}, but to deal with the richer set of rotations---which is now a continuous group---each $Z$, $X$, and $W$ is promoted to a two-qubit operator: 
$Z_i\mapsto \bar{Z}_i \equiv Z_{2i-1}Z_{2i}$, $X_i\mapsto \bar{X}_i \equiv X_{2i-1}X_{2i}$, $W_i^{{\alpha}} \mapsto \bar{W}_i^{{\alpha}}\equiv W_{2i-1}^{{\alpha}}  \otimes W_{2i}^{{\alpha}}$, giving rise to $6$-local Hamiltonians. Thus Eq.~\eqref{eq:Hijk111} becomes
\bea
\bar{H}_{ijk}^{(111)} =\bar{Z}_i \bar{Z}_j \bar{Z}_k &-& 3 (\bar{Z}_i +\bar{Z}_j +\bar{Z}_k)\nonumber \\
&-&(\bar{Z}_i \bar{Z}_j+ \bar{Z}_i \bar{Z}_k+\bar{Z}_j \bar{Z}_k) \ .
\label{eq:H'ijk111}
\eea
Let $\bar{W}(x)\equiv\bigotimes_{i=1}^{n} \bar{W}_i^{x_i}$. It is again straightforward to check that $\bar{W}(x) \bar{H}_{ijk}^{(111)}\bar{W}^{\dagger}(x)$
is stoquastic for any combination of the binary variables $(x_i,x_j,x_k)$ except for $(1,1,1)$.

Likewise, let $H_{ijk}^{(\alpha\beta\gamma)}\mapsto \bar{H}_{ijk}^{(\alpha\beta\gamma)}$ and ${H}_{\mathrm{3SAT}}\mapsto \bar{H}_{\mathrm{3SAT}}= \sum_{C} \bar{H}_{ijk}^{(\alpha\beta\gamma)}$ [generalizing Eqs.~\eqref{eq:Hijkabc} and \eqref{eq:H3SAT}], where $C$ again denotes the corresponding set of clauses in a 3SAT instance, constructed just as in the proof of Lemma~\ref{lemmaIW}. This, again, establishes a bijection between 3SAT clauses and `3SAT-clause Hamiltonians,' now of the form $\bar{H}_{\mathrm{3SAT}}$. 

In lieu of Eq.~\eqref{eq:tildeH} we consider
\bea
\tilde{\bar{H}}_{\mathrm{3SAT}}=\bar{H}_{\mathrm{3SAT}} + c  H'_{0} \ , \quad H'_{0}=  -\sum_{i=1}^{n} 2 \bar{Z}_i + \bar{X}_i\ ,
\label{eq:tildebarH3SAT}
\eea
where $c = O(1)$. 
Just as in the proof of Theorem~\ref{thm1} we prove that (i) any satisfying assignment of a 3SAT instance provides a curing $Q$ for the corresponding Hamiltonian $\tilde{\bar{H}}_{\mathrm{3SAT}}$, and (ii) any $Q$ that cures $\tilde{\bar{H}}_{\mathrm{3SAT}}$ provides a satisfying assignment for the corresponding 3SAT instance. Because of the relation between a single-qubit orthogonal matrix and a single qubit rotation it suffices to prove the hardness only for pure rotations (see the SI, Sec.~\ref{SI4-a}); we let $R(\theta_i)=
\begin{bmatrix}
\cos{\theta_i} & - \sin{\theta_i}\\ 
 \sin{\theta_i} &  \cos{\theta_i}
\end{bmatrix}$ denote a rotation by angle $\theta_i$.

(i) Let $P(x)$ denote the product of $2n$ single-qubit rotations such that if $x_i=0$ then qubits $2i-1$ and $2i$ are unchanged, or if $x_i=1$ then they are both rotated by $R(\frac{\pi}{4})=XW$. Just as in the proof of Theorem~\ref{thm1}, if the 3SAT instance has a satisfying assignment $x^*$ then $P(x^*) \tilde{\bar{H}}_{\mathrm{3SAT}} P^T(x^*)$ is stoquastic (see the SI, Sec.~\ref{SI4-b}). 

(ii) We need to prove that any rotation that cures $\tilde{\bar{H}}_{\mathrm{3SAT}}$ provides a satisfying assignment for the corresponding 3SAT instance. We do this in two steps: (a) In the SI (Sec.~\ref{SI4-c}) we prove a lemma showing that for any $\tilde{\bar{H}}_{\mathrm{3SAT}}$, any curing rotation $R=\bigotimes_{i=1}^{2n} R(\theta_i)$ has to satisfy the condition that $(\theta_{2i-1},\theta_{2i}) \in \{ (\frac{\pi}{2},\frac{\pi}{2}),(\frac{\pi}{4},\frac{\pi}{4}),(0,0), (\frac{-\pi}{4},\frac{-\pi}{4})\}$ $ \forall i$. This is the crucial step, since it reduces the problem from a continuum of angles to a discrete set. If $(\theta_{2i-1},\theta_{2i})\in\{(0,0),(\frac{\pi}{2},\frac{\pi}{2})\}$ we assign $x_i=0$, while if $(\theta_{2i-1},\theta_{2i})\in\{(-\frac{\pi}{4},-\frac{\pi}{4}),(\frac{\pi}{4},\frac{\pi}{4})\}$ we assign $x_i=1$, since rotations with the angles in each pair have the same effect. (b) If $R$ cures $\tilde{\bar{H}}_{\mathrm{3SAT}}$, $x=\{x_i\}$ satisfies the corresponding 3SAT instance. 
To see this we first note that if $R$ cures $\tilde{\bar{H}}_{\mathrm{3SAT}}$ it must cure all the clauses separately: Using step (a), we know that any such solution must be one of the four possible cases. Therefore, if $R$ were to cure $\tilde{\bar{H}}_{\mathrm{3SAT}}$ but does not cure one of the 3SAT-clause Hamiltonians, it would result in a $\bar{X}_{i}\bar{X}_{j}\bar{X}_{k}$ term in the corresponding clause. Since no other 3SAT-clause Hamiltonian in $\tilde{\bar{H}}_{\mathrm{3SAT}}$ contains an identical $\bar{X}_{i}\bar{X}_{j}\bar{X}_{k}$ term, these positive off-diagonal elements cannot be cancelled out or made negative regardless of the choice of the other variables in the assignment. Therefore if $R$ cures $\tilde{\bar{H}}_{\mathrm{3SAT}}$ it also necessarily separately cures all the terms in $\tilde{\bar{H}}_{\mathrm{3SAT}}$. By construction, if $R$ cures a term $\bar{H}_{ijk}^{(\alpha\beta\gamma)}$, the string $x$ satisfies the corresponding 3SAT clause $C_{ijk}^{(\alpha\beta\gamma)}$. Thus $x$ satisfies  all the clauses in the corresponding 3SAT instance.
The decision problem for the existence of $R$ (and hence $Q$) is therefore NP-hard. Given a unitary $U$ and a set of local terms $\{H_a\}$, verifying whether $U$ cures all of the terms is clearly efficient and therefore this problem is NP-complete.
\end{proof}

\textit{Implications}.---
An immediate and striking implication of Theorem~\ref{thm1} is that even under the promise that a non-stoquastic Hamiltonian can be cured by one-local Clifford unitaries (corresponding to trivial basis changes), the problem of actually finding this transformation is unlikely to have a polynomial-time solution.

An interesting implication of Theorem~\ref{thm2} 
is the possibility of constructing `secretly stoquastic' Hamiltonians.
I.e., one may generate stoquastic quantum many-body Hamiltonians $H_{\mathrm{stoq}}$, but present these in a `scrambled' non-stoquastic form $H_{\mathrm{nonstoq}}=U H_{\mathrm{stoq}} U^{\dagger}$ where $U$ is a tensor product of single-qubit orthogonal matrices (or in the general case a constant depth quantum circuit). We conjecture that the latter Hamiltonians will be computationally hard to simulate using QMC by parties that have no access to the `descrambling' circuit $U$.
In other words, it is possible to generate efficiently simulable spin models that might be inefficient to simulate unless one has access to the `secret key' to make them stoquastic. In the SI (Sec.~\ref{SI5}) we note that this observation may potentially have cryptographic applications.


Our work also has implications for the connection between the sign problem and the NP-hardness of a QMC simulation. A prevailing view of this issue associates the origin of the NP-hardness of a QMC simulation to the relation between a (`fermionic') Hamiltonian that suffers from a sign problem and the corresponding (`bosonic') Hamiltonian obtained by replacing every coupling coefficient by its absolute value (e.g.,  $\sum_{ij} {J}_{ij} X_iX_j \mapsto \sum_{ij} |J_{ij} | X_i X_j$)
\footnote{For example in Ref.~\protect\cite{Troyer2005} a solution to sign problem is defined as: \protect\textquotedblleft an algorithm of polynomial complexity to evaluate the thermal average $\protect\langle A \protect\rangle$\protect\textquotedblright ... \protect\textquotedblleft [f]or a quantum system that suffers from a sign problem for an observable $A$, and for which there exists a polynomial complexity algorithm for the related bosonic system\protect\textquotedblright.}.
The view we advocate here is that a solution to the sign problem is to find an efficiently computable curing transformation that removes it but has the same physics (in general the fermionic and bosonic versions of the same Hamiltonian do not), i.e., conserves thermal averages \footnote{
Consider the following example. (i) $H_X = \protect\sum_{ij} {J}_{ij} X_iX_j$, with ${J}_{ij}$ randomly chosen from the set $\protect\{0,\protect\pm J\protect\}$ on a three-dimensional lattice, has a sign problem. (ii) Deciding whether its ground state energy is below a given bound is NP-complete~\cite{Barahona1982}. (iii) Deciding the same for its bosonic and sign-problem-free version $H_{|X|} = \protect\sum_{ij} |J_{ij}| X_i X_j$ is in BPP (classical polynomial time with bounded error) since this Hamiltonian is that of a simple ferromagnet. The conclusion drawn in Ref.~\cite{Troyer2005} was that since the bosonic version is easy to simulate, the sign problem is the origin of the NP-hardness of a QMC simulation of this model ($H_X$). However, note that 
computing thermal averages via a QMC simulation of $H_X$ is the same as for $H_Z =W^{\protect\otimes n} H_{X} W^{\protect\otimes n}= \protect\sum_{ij} {J}_{ij} Z_i Z_j$, which is stoquastic and has no sign problem. Thus, the sign problem of $H_X$ is efficiently curable, after which (when it is presented as $H_Z$) deciding its ground state energy remains NP-hard.}.

\textit{Conclusions and open questions}.---%
We have proposed a new definition of stoquasticity (or absence of the sign problem) 
of quantum many-body Hamiltonians that is motivated by computational complexity considerations. We discussed the circumstances under which non-stoquastic Hamiltonians can in fact be made stoquastic by the application of single-qubit rotations and in turn potentially become efficiently simulable by QMC algorithms. We have demonstrated that finding the required rotations is computationally hard when they are restricted to the one-qubit Clifford group or one-qubit continuous orthogonal matrices.

These results raise multiple questions of interest.
It is important to clarify the computational complexity of finding the curing transformation in the case of constant-depth circuits that also allow two-body rotations, whether discrete or continuous. Also, since our NP-completeness proof involved $3$ and $6$-local Hamiltonians, it is interesting to try to reduce it to $2$-local building blocks. Another direction into which these results can be extended is to relax the constraints on the off-diagonal elements and require that they are smaller than some small $\epsilon>0$. This is relevant when some small positive off-diagonal elements can be ignored in a QMC simulation. 
 

Finally, it is natural to reconsider our results from the perspective of  quantum computing. Namely, for non-stoquastic Hamiltonians that are curable, do there exist quantum algorithms that cure the sign problem more efficiently than is possible classically? 
With the advent of quantum computers, specifically quantum annealers, it may be the case that these can be used as quantum simulators, and as such they will not be plagued by the sign problem. Will such physical implementations of quantum computers offer advantages over classical computing even for problems that are incurably non-stoquastic?
We leave these as open questions to be addressed in future studies.

\textit{Acknowledgments}.---
The research is based upon work (partially) supported by the Office of
the Director of National Intelligence (ODNI), Intelligence Advanced
Research Projects Activity (IARPA), via the U.S. Army Research Office
contract W911NF-17-C-0050. The views and conclusions contained herein are those of the authors and should not be interpreted as necessarily representing the official policies or endorsements, either expressed or implied, of the ODNI, IARPA, or the U.S. Government. The U.S. Government is authorized to reproduce and distribute reprints for Governmental purposes notwithstanding any copyright annotation thereon.
We thank Ehsan Emamjomeh-Zadeh, Iman Marvian, Evgeny Mozgunov, Ben Reichardt, and Federico Spedalieri for useful discussions.

\bibliography{refs}

\begin{thebibliography}{23}%
\makeatletter
\providecommand \@ifxundefined [1]{%
 \@ifx{#1\undefined}
}%
\providecommand \@ifnum [1]{%
 \ifnum #1\expandafter \@firstoftwo
 \else \expandafter \@secondoftwo
 \fi
}%
\providecommand \@ifx [1]{%
 \ifx #1\expandafter \@firstoftwo
 \else \expandafter \@secondoftwo
 \fi
}%
\providecommand \natexlab [1]{#1}%
\providecommand \enquote  [1]{``#1''}%
\providecommand \bibnamefont  [1]{#1}%
\providecommand \bibfnamefont [1]{#1}%
\providecommand \citenamefont [1]{#1}%
\providecommand \href@noop [0]{\@secondoftwo}%
\providecommand \href [0]{\begingroup \@sanitize@url \@href}%
\providecommand \@href[1]{\@@startlink{#1}\@@href}%
\providecommand \@@href[1]{\endgroup#1\@@endlink}%
\providecommand \@sanitize@url [0]{\catcode `\\12\catcode `\$12\catcode
  `\&12\catcode `\#12\catcode `\^12\catcode `\_12\catcode `\%12\relax}%
\providecommand \@@startlink[1]{}%
\providecommand \@@endlink[0]{}%
\providecommand \url  [0]{\begingroup\@sanitize@url \@url }%
\providecommand \@url [1]{\endgroup\@href {#1}{\urlprefix }}%
\providecommand \urlprefix  [0]{URL }%
\providecommand \Eprint [0]{\href }%
\providecommand \doibase [0]{http://dx.doi.org/}%
\providecommand \selectlanguage [0]{\@gobble}%
\providecommand \bibinfo  [0]{\@secondoftwo}%
\providecommand \bibfield  [0]{\@secondoftwo}%
\providecommand \translation [1]{[#1]}%
\providecommand \BibitemOpen [0]{}%
\providecommand \bibitemStop [0]{}%
\providecommand \bibitemNoStop [0]{.\EOS\space}%
\providecommand \EOS [0]{\spacefactor3000\relax}%
\providecommand \BibitemShut  [1]{\csname bibitem#1\endcsname}%
\let\auto@bib@innerbib\@empty
\bibitem [{\citenamefont {Loh}\ \emph {et~al.}(1990)\citenamefont {Loh},
  \citenamefont {Gubernatis}, \citenamefont {Scalettar}, \citenamefont {White},
  \citenamefont {Scalapino},\ and\ \citenamefont {Sugar}}]{Loh-PRB-90}%
  \BibitemOpen
  \bibfield  {author} {\bibinfo {author} {\bibfnamefont {E.~Y.}\ \bibnamefont
  {Loh}}, \bibinfo {author} {\bibfnamefont {J.~E.}\ \bibnamefont {Gubernatis}},
  \bibinfo {author} {\bibfnamefont {R.~T.}\ \bibnamefont {Scalettar}}, \bibinfo
  {author} {\bibfnamefont {S.~R.}\ \bibnamefont {White}}, \bibinfo {author}
  {\bibfnamefont {D.~J.}\ \bibnamefont {Scalapino}}, \ and\ \bibinfo {author}
  {\bibfnamefont {R.~L.}\ \bibnamefont {Sugar}},\ }\href {\doibase
  10.1103/PhysRevB.41.9301} {\bibfield  {journal} {\bibinfo  {journal} {Phys.
  Rev. B}\ }\textbf {\bibinfo {volume} {41}},\ \bibinfo {pages} {9301}
  (\bibinfo {year} {1990})}\BibitemShut {NoStop}%
\bibitem [{\citenamefont {Landau}\ and\ \citenamefont
  {Binder}(2005)}]{Landau:2005:GMC:1051461}%
  \BibitemOpen
  \bibfield  {author} {\bibinfo {author} {\bibfnamefont {D.}~\bibnamefont
  {Landau}}\ and\ \bibinfo {author} {\bibfnamefont {K.}~\bibnamefont
  {Binder}},\ }\href@noop {} {\emph {\bibinfo {title} {A Guide to Monte Carlo
  Simulations in Statistical Physics}}}\ (\bibinfo  {publisher} {Cambridge
  University Press},\ \bibinfo {address} {New York, NY, USA},\ \bibinfo {year}
  {2005})\BibitemShut {NoStop}%
\bibitem [{\citenamefont {Barkema}(1999)}]{newman}%
  \BibitemOpen
  \bibfield  {author} {\bibinfo {author} {\bibfnamefont {M.~N. .~G.}\
  \bibnamefont {Barkema}},\ }\href@noop {} {\emph {\bibinfo {title} {Monte
  Carlo Methods in Statistical Physics}}}\ (\bibinfo  {publisher} {Oxford
  Uinversity Press},\ \bibinfo {year} {1999})\BibitemShut {NoStop}%
\bibitem [{\citenamefont {Bravyi}\ \emph {et~al.}(2008)\citenamefont {Bravyi},
  \citenamefont {DiVincenzo}, \citenamefont {Oliveira},\ and\ \citenamefont
  {Terhal}}]{Bravyi:QIC08}%
  \BibitemOpen
  \bibfield  {author} {\bibinfo {author} {\bibfnamefont {S.}~\bibnamefont
  {Bravyi}}, \bibinfo {author} {\bibfnamefont {D.~P.}\ \bibnamefont
  {DiVincenzo}}, \bibinfo {author} {\bibfnamefont {R.~I.}\ \bibnamefont
  {Oliveira}}, \ and\ \bibinfo {author} {\bibfnamefont {B.~M.}\ \bibnamefont
  {Terhal}},\ }\href {https://arxiv.org/abs/quant-ph/0606140} {\bibfield
  {journal} {\bibinfo  {journal} {Quant. Inf. Comp.}\ }\textbf {\bibinfo
  {volume} {8}},\ \bibinfo {pages} {0361} (\bibinfo {year} {2008})}\BibitemShut
  {NoStop}%
\bibitem [{\citenamefont {Bravyi}\ and\ \citenamefont
  {Terhal}(2009)}]{Bravyi:2009sp}%
  \BibitemOpen
  \bibfield  {author} {\bibinfo {author} {\bibfnamefont {S.}~\bibnamefont
  {Bravyi}}\ and\ \bibinfo {author} {\bibfnamefont {B.}~\bibnamefont
  {Terhal}},\ }\bibfield  {booktitle} {\emph {\bibinfo {booktitle} {SIAM
  Journal on Computing}},\ }\href {\doibase 10.1137/08072689X} {\bibfield
  {journal} {\bibinfo  {journal} {SIAM Journal on Computing}\ }\textbf
  {\bibinfo {volume} {39}},\ \bibinfo {pages} {1462} (\bibinfo {year}
  {2009})}\BibitemShut {NoStop}%
\bibitem [{Note1()}]{Note1}%
  \BibitemOpen
  \bibinfo {note} {E.g., in the path-integral formulation of QMC with respect
  to a basis ${\protect \cal {B}}=\protect \{ b \protect \}$, the partition
  function $Z$ is reduced to an $L$-fold product of sums over complete sets of
  basis states, $\protect \{b_1\protect \},\protect \ldots ,\protect
  \{b_L\protect \}$, which are weighted by the size of the imaginary-time slice
  $\Delta \tau =\beta /L$ and the off-diagonal matrix elements of $H$. Namely,
  $Z \approx \DOTSB \prod@ \slimits@ _{l=1}^{L} \DOTSB \sum@ \slimits@ _{b_l}
  \mathinner {\delimiter "426830A {b_l}|} \protect \mathrm {e}^{-\Delta \tau
  H_{l,l+1}}\mathinner {|{b_{l+1}}\delimiter "526930B }$, where $L$ is the
  number of slices and periodic boundary conditions are assumed. The connection
  to stoquasticity is that when all the off-diagonal matrix elements,
  $H_{j,j+1}$ in the given basis are non-positive, these weights are purely
  positive for each time slice.}\BibitemShut {Stop}%
\bibitem [{\citenamefont {Troyer}\ and\ \citenamefont
  {Wiese}(2005)}]{Troyer2005}%
  \BibitemOpen
  \bibfield  {author} {\bibinfo {author} {\bibfnamefont {M.}~\bibnamefont
  {Troyer}}\ and\ \bibinfo {author} {\bibfnamefont {U.-J.}\ \bibnamefont
  {Wiese}},\ }\href {\doibase 10.1103/PhysRevLett.94.170201} {\bibfield
  {journal} {\bibinfo  {journal} {Phys. Rev. Lett.}\ }\textbf {\bibinfo
  {volume} {94}},\ \bibinfo {pages} {170201} (\bibinfo {year}
  {2005})}\BibitemShut {NoStop}%
\bibitem [{\citenamefont {Cubitt}\ and\ \citenamefont
  {Montanaro}(2016)}]{Cubitt:2016vl}%
  \BibitemOpen
  \bibfield  {author} {\bibinfo {author} {\bibfnamefont {T.}~\bibnamefont
  {Cubitt}}\ and\ \bibinfo {author} {\bibfnamefont {A.}~\bibnamefont
  {Montanaro}},\ }\bibfield  {booktitle} {\emph {\bibinfo {booktitle} {SIAM
  Journal on Computing}},\ }\href {\doibase 10.1137/140998287} {\bibfield
  {journal} {\bibinfo  {journal} {SIAM Journal on Computing}\ }\textbf
  {\bibinfo {volume} {45}},\ \bibinfo {pages} {268} (\bibinfo {year}
  {2016})}\BibitemShut {NoStop}%
\bibitem [{Note2()}]{Note2}%
  \BibitemOpen
  \bibinfo {note} {Throughout this work we reserve the term `efficient' to mean
  that the algorithm requires at most a polynomial run-time in the problem size
  (the number of variables). For example, an algorithm equilibrates efficiently
  if it can correctly samples from the Gibbs distribution of the Hamiltonian in
  question with at most a polynomial run-time in the problem size and the given
  statistical error [see part (iii) of the definition of the sign problem given
  in Ref.~\cite {Troyer2005} for a precise statement.]}\BibitemShut {NoStop}%
\bibitem [{\citenamefont {Barahona}(1982)}]{Barahona1982}%
  \BibitemOpen
  \bibfield  {author} {\bibinfo {author} {\bibfnamefont {F.}~\bibnamefont
  {Barahona}},\ }\href {http://stacks.iop.org/0305-4470/15/i=10/a=028}
  {\bibfield  {journal} {\bibinfo  {journal} {J. Phys. A: Math. Gen}\ }\textbf
  {\bibinfo {volume} {15}},\ \bibinfo {pages} {3241} (\bibinfo {year}
  {1982})}\BibitemShut {NoStop}%
\bibitem [{\citenamefont {Okunishi}\ and\ \citenamefont
  {Harada}(2014)}]{PhysRevB.89.134422}%
  \BibitemOpen
  \bibfield  {author} {\bibinfo {author} {\bibfnamefont {K.}~\bibnamefont
  {Okunishi}}\ and\ \bibinfo {author} {\bibfnamefont {K.}~\bibnamefont
  {Harada}},\ }\href {\doibase 10.1103/PhysRevB.89.134422} {\bibfield
  {journal} {\bibinfo  {journal} {Phys. Rev. B}\ }\textbf {\bibinfo {volume}
  {89}},\ \bibinfo {pages} {134422} (\bibinfo {year} {2014})}\BibitemShut
  {NoStop}%
\bibitem [{\citenamefont {Alet}\ \emph {et~al.}(2016)\citenamefont {Alet},
  \citenamefont {Damle},\ and\ \citenamefont {Pujari}}]{Alet}%
  \BibitemOpen
  \bibfield  {author} {\bibinfo {author} {\bibfnamefont {F.}~\bibnamefont
  {Alet}}, \bibinfo {author} {\bibfnamefont {K.}~\bibnamefont {Damle}}, \ and\
  \bibinfo {author} {\bibfnamefont {S.}~\bibnamefont {Pujari}},\ }\href
  {\doibase 10.1103/PhysRevLett.117.197203} {\bibfield  {journal} {\bibinfo
  {journal} {Phys. Rev. Lett.}\ }\textbf {\bibinfo {volume} {117}},\ \bibinfo
  {pages} {197203} (\bibinfo {year} {2016})}\BibitemShut {NoStop}%
\bibitem [{\citenamefont {Honecker}\ \emph {et~al.}(2016)\citenamefont
  {Honecker}, \citenamefont {Wessel}, \citenamefont {Kerkdyk}, \citenamefont
  {Pruschke}, \citenamefont {Mila},\ and\ \citenamefont
  {Normand}}]{PhysRevB.93.054408}%
  \BibitemOpen
  \bibfield  {author} {\bibinfo {author} {\bibfnamefont {A.}~\bibnamefont
  {Honecker}}, \bibinfo {author} {\bibfnamefont {S.}~\bibnamefont {Wessel}},
  \bibinfo {author} {\bibfnamefont {R.}~\bibnamefont {Kerkdyk}}, \bibinfo
  {author} {\bibfnamefont {T.}~\bibnamefont {Pruschke}}, \bibinfo {author}
  {\bibfnamefont {F.}~\bibnamefont {Mila}}, \ and\ \bibinfo {author}
  {\bibfnamefont {B.}~\bibnamefont {Normand}},\ }\href {\doibase
  10.1103/PhysRevB.93.054408} {\bibfield  {journal} {\bibinfo  {journal} {Phys.
  Rev. B}\ }\textbf {\bibinfo {volume} {93}},\ \bibinfo {pages} {054408}
  (\bibinfo {year} {2016})}\BibitemShut {NoStop}%
\bibitem [{Note3()}]{Note3}%
  \BibitemOpen
  \bibinfo {note} {{The thermal average is $\protect \langle A \protect \rangle
  \protect \equiv \protect \frac {1}{Z}\protect \mathrm {Tr}(e^{-\beta H} A)$
  where $\protect \beta $ is the inverse temperature, and our claim is that
  $\protect \langle A \protect \rangle = \protect \langle A' \protect \rangle $
  where $A' = UAU^\dagger $ for unitary $U$. Here is the proof. First,
  $Z=\protect \mathrm {Tr}(e^{-\beta H})=\protect \mathrm {Tr}( U e^{ -\beta H
  } U^\dagger ) = \protect \mathrm {Tr}( e^{ -\beta H'} ) = Z'$ where $H' =
  UHU^\dagger $. Thus $\protect \langle A \protect \rangle = \protect \frac
  {1}{Z}\protect \mathrm {Tr}(U e^{-\beta H} U^\dagger UA U^\dagger ) =
  \protect \frac {1}{Z'}\protect \mathrm {Tr}(e^{-\beta H'} A') = \protect
  \langle A' \protect \rangle $.}}\BibitemShut {Stop}%
\bibitem [{Note4()}]{Note4}%
  \BibitemOpen
  \bibinfo {note} {A related approach was discussed by Barbara Terhal at the
  AQC17 conference~\protect \url
  {http://www.smapip.is.tohoku.ac.jp/~aqc2017/program.html}.}\BibitemShut
  {Stop}%
\bibitem [{\citenamefont {Hastings}(2016)}]{hastings2016quantum}%
  \BibitemOpen
  \bibfield  {author} {\bibinfo {author} {\bibfnamefont {M.}~\bibnamefont
  {Hastings}},\ }\href@noop {} {\bibfield  {journal} {\bibinfo  {journal}
  {Journal of Mathematical Physics}\ }\textbf {\bibinfo {volume} {57}},\
  \bibinfo {pages} {015210} (\bibinfo {year} {2016})}\BibitemShut {NoStop}%
\bibitem [{\citenamefont {Nielsen}\ and\ \citenamefont
  {Chuang}(2010)}]{nielsen2010quantum}%
  \BibitemOpen
  \bibfield  {author} {\bibinfo {author} {\bibfnamefont {M.~A.}\ \bibnamefont
  {Nielsen}}\ and\ \bibinfo {author} {\bibfnamefont {I.~L.}\ \bibnamefont
  {Chuang}},\ }\href@noop {} {\emph {\bibinfo {title} {Quantum computation and
  quantum information}}}\ (\bibinfo  {publisher} {{Cambridge University
  Press}},\ \bibinfo {year} {2010})\BibitemShut {NoStop}%
\bibitem [{\citenamefont {Karp}(1972)}]{Karp:21-problems}%
  \BibitemOpen
  \bibfield  {author} {\bibinfo {author} {\bibfnamefont {R.}~\bibnamefont
  {Karp}},\ }in\ \href
  {http://link.springer.com/chapter/10.1007/978-1-4684-2001-2_9} {\emph
  {\bibinfo {booktitle} {Complexity of Computer Computations}}},\ \bibinfo
  {series and number} {The IBM Research Symposia Series},\ \bibinfo {editor}
  {edited by\ \bibinfo {editor} {\bibfnamefont {R.~E.}\ \bibnamefont {Miller}}\
  and\ \bibinfo {editor} {\bibfnamefont {J.~W.}\ \bibnamefont {Thatcher}}}\
  (\bibinfo  {publisher} {Plenum},\ \bibinfo {address} {New York},\ \bibinfo
  {year} {1972})\ Chap.~\bibinfo {chapter} {9}, p.~\bibinfo {pages}
  {85}\BibitemShut {NoStop}%
\bibitem [{Note5()}]{Note5}%
  \BibitemOpen
  \bibinfo {note} {In general, $M$ can grow polynomially (usually linearly for
  hard instances) with $n$, but $\protect \mathaccentV {tilde}07E{H}_{\protect
  \mathrm {3SAT}}$ remain $k$-local since the operator norm of each of its
  terms is $\protect \textrm {poly}(n)$, as required by the definition of a
  $k$-local Hamiltonian (see, e.g., Definition~1 of Ref.~\cite
  {Cubitt:2016vl}).}\BibitemShut {Stop}%
\bibitem [{Note6()}]{Note6}%
  \BibitemOpen
  \bibinfo {note} {For example in Ref.~\protect \cite {Troyer2005} a solution
  to sign problem is defined as: \protect \textquotedblleft an algorithm of
  polynomial complexity to evaluate the thermal average $\protect \langle A
  \protect \rangle $\protect \textquotedblright ... \protect \textquotedblleft
  [f]or a quantum system that suffers from a sign problem for an observable
  $A$, and for which there exists a polynomial complexity algorithm for the
  related bosonic system\protect \textquotedblright .}\BibitemShut {Stop}%
\bibitem [{Note7()}]{Note7}%
  \BibitemOpen
  \bibinfo {note} {Consider the following example. (i) $H_X = \protect \sum
  _{ij} {J}_{ij} X_iX_j$, with ${J}_{ij}$ randomly chosen from the set
  $\protect \{0,\protect \pm J\protect \}$ on a three-dimensional lattice, has
  a sign problem. (ii) Deciding whether its ground state energy is below a
  given bound is NP-complete~\cite {Barahona1982}. (iii) Deciding the same for
  its bosonic and sign-problem-free version $H_{|X|} = \protect \sum _{ij}
  |J_{ij}| X_i X_j$ is in BPP (classical polynomial time with bounded error)
  since this Hamiltonian is that of a simple ferromagnet. The conclusion drawn
  in Ref.~\cite {Troyer2005} was that since the bosonic version is easy to
  simulate, the sign problem is the origin of the NP-hardness of a QMC
  simulation of this model ($H_X$). However, note that computing thermal
  averages via a QMC simulation of $H_X$ is the same as for $H_Z =W^{\protect
  \otimes n} H_{X} W^{\protect \otimes n}= \protect \sum _{ij} {J}_{ij} Z_i
  Z_j$, which is stoquastic and has no sign problem. Thus, the sign problem of
  $H_X$ is efficiently curable, after which (when it is presented as $H_Z$)
  deciding its ground state energy remains NP-hard.}\BibitemShut {Stop}%
\bibitem [{\citenamefont {Hen}\ \emph {et~al.}(2015)\citenamefont {Hen},
  \citenamefont {Job}, \citenamefont {Albash}, \citenamefont {Ronnow},
  \citenamefont {Troyer},\ and\ \citenamefont {Lidar}}]{Hen:2015rt}%
  \BibitemOpen
  \bibfield  {author} {\bibinfo {author} {\bibfnamefont {I.}~\bibnamefont
  {Hen}}, \bibinfo {author} {\bibfnamefont {J.}~\bibnamefont {Job}}, \bibinfo
  {author} {\bibfnamefont {T.}~\bibnamefont {Albash}}, \bibinfo {author}
  {\bibfnamefont {T.~F.}\ \bibnamefont {Ronnow}}, \bibinfo {author}
  {\bibfnamefont {M.}~\bibnamefont {Troyer}}, \ and\ \bibinfo {author}
  {\bibfnamefont {D.~A.}\ \bibnamefont {Lidar}},\ }\href
  {http://link.aps.org/doi/10.1103/PhysRevA.92.042325} {\bibfield  {journal}
  {\bibinfo  {journal} {{Phys. Rev. A}}\ }\textbf {\bibinfo {volume} {92}},\
  \bibinfo {pages} {042325} (\bibinfo {year} {2015})}\BibitemShut {NoStop}%
\bibitem [{\citenamefont {Hamze}\ \emph {et~al.}(2017)\citenamefont {Hamze},
  \citenamefont {Jacob}, \citenamefont {Ochoa}, \citenamefont {Wang},\ and\
  \citenamefont {Katzgraber}}]{Hamze:2017aa}%
  \BibitemOpen
  \bibfield  {author} {\bibinfo {author} {\bibfnamefont {F.}~\bibnamefont
  {Hamze}}, \bibinfo {author} {\bibfnamefont {D.~C.}\ \bibnamefont {Jacob}},
  \bibinfo {author} {\bibfnamefont {A.~J.}\ \bibnamefont {Ochoa}}, \bibinfo
  {author} {\bibfnamefont {W.}~\bibnamefont {Wang}}, \ and\ \bibinfo {author}
  {\bibfnamefont {H.~G.}\ \bibnamefont {Katzgraber}},\ }\href
  {http://arXiv.org/abs/1711.04083} {\bibfield  {journal} {\bibinfo  {journal}
  {arXiv:1711.04083}\ } (\bibinfo {year} {2017})}\BibitemShut {NoStop}%
\end{thebibliography}%

\onecolumngrid

\newpage

\appendix
\section*{Supplemental Information}

\section{Grouping terms (without changing the basis)} 
\label{SI1}

As discussed in the main text, one ambiguity in the definition of stoquastic Hamiltonians is in the choice of the set $\{H_a\}$. With this motivation, and ignoring the freedom in choosing a basis, we address the following question.

\textit{Problem:} We are given the $k$-local $H=\sum_a{H_a}$, i.e., each $H_a$ acts nontrivially on at most $k$ qubits. In the same basis (without any rotation), find a new set $\{H'_a\}$ satisfying $H=\sum_a{H'_a}$, where each $H'_a$ is $k'$-local and stoquastic (if such a set exists).

Obviously, if the total Hamiltonian is stoquastic then considering the total Hamiltonian as one single Hamiltonian is a valid solution with $k'=n$. This description of the Hamiltonian requires a $2^n \times 2^n$ matrix. We would prefer a $k'$-local Hamiltonian, i.e., a set consisting of a polynomial number of terms, each $2^{k'} \times 2^{k'}$ where $k'$ is a constant independent of $n$.

\textit{Solution:} One simple strategy is to consider any $k'$-local combination of qubits, and to try to find a grouping that makes all of these ${n \choose k'}$ terms stoquastic. To do so, for any $k'$-local combination of qubits we generate a set of inequalities. First, for a fixed combination of qubits, we add the terms in $H=\sum_a{H_a}$ that act nontrivially only on those $k'$ qubits, each with an unknown weight that will be determined later. Then we write down the conditions on the weights to ensure that all the off-diagonal elements are non-positive. This is done for all the ${n \choose k'}$ combinations to get the complete set of linear inequalities.  By this procedure, the problem reduces to finding a feasible point for this set of linear inequalities, which can be solved efficiently. (In practice one can use linear programming optimization tools to check whether such a feasible point exists.) When there is no feasible point for a specific value of $k'$ we can increase the value of $k'$ and search again.
 
\textit{Example:} Assume we are given $H=Z_1 X_2 -2 X_2+ X_2 Z_3$ and the goal is to find a stoquastic description with $k'=2$. We combine the terms acting on qubits $1$ and $2$ and then the terms acting on qubits $2$ and $3$ (there is no term on qubits $1$ and $3$). We construct $h_{1,2}= \alpha_1 Z_1 X_2 + \alpha_2 (-2 X_2)$ and $h_{2,3}= \alpha_3 (-2 X_2)+ \alpha_4 X_2 Z_3$. There are two types of constraints: (1) constraints enforcing $H=h_{1,2}+h_{2,3}$:
\bea
\alpha_1=\alpha_4 =1, \alpha_2+\alpha_3=1\ ,
\eea
and (2) constraints from stoquasticity of each of the two Hamiltonians: 
\bea
1-2\alpha_2 \leq 0, -1-2\alpha_2 \leq 0 ; \\1-2\alpha_3 \leq 0, -1-2\alpha_3 \leq 0.
\eea
Simplifying these inequalities we have $.5 \leq \alpha_2,\alpha_3$ and $\alpha_2+\alpha_3=1$, which clearly has only one feasible point: $\alpha_2=\alpha_3=.5$. The corresponding terms are : $H'_1=h_{1,2}=Z_1 X_2 - X_2$ and $H'_2=h_{2,3}=- X_2+ X_2 Z_3$. Both of these terms are stoquastic and they satisfy $H=\sum_a{H'_a}$.

\section{Curing using Pauli operators}
\label{SI:Pauli}

In Appendix~\ref{SI3} below, we show that conjugating a Hamiltonian by a tensor product of Pauli $X$ operators or identity operators only shuffles the off-diagonal elements without changing their values.
Recalling that $Y=iXZ$, we thus conclude that choosing between Pauli operators to cure a Hamiltonian is equivalent to choosing between $I$ and $Z$ operators. Therefore, given local terms of a $k$-local Hamiltonian $\{H_a\}$ as input, the goal is to find a string $x=(x_1,\dots,x_n)$ such that $U=\otimes_{i=1}^n Z^{x_i}$ cures each of the local terms $\{H_a\}$ separately.

Each multi-qubit Pauli operator in $H_a$ can be decomposed into $X$ components and $Z$ components. We group all the terms in each $H_a$ that share the same $X$ component. For example if $H_a$ includes $Y_1Y_2,3X_1X_2,$ and $X_1X_2Z_3$, we combine them into one single term $X_1 X_2 (-Z_1 Z_2+3+Z_3)$. Conjugating this term with $U$ yields $(-1)^{x_1+ x_2} X_1 X_2 (-Z_1 Z_2+3+Z_3)$. As terms with different $X$ components do not correspond to overlapping off-diagonal elements in $H_a$, the combined $Z$ part fixes a constraint on $\{x_i\}$ based on the positivity or negativity of all its elements (if the combined $Z$ part has both positive and negative elements we conclude that there is no $U$ that can cure the input $H$). In this example, $(-1)^{x_1+ x_2} X_1 X_2 (-Z_1 Z_2+3+Z_3)$ becomes stoquastic iff $x_1+ x_2 \equiv 1 \mod 2$. 

We combine all these linear equations in mod $2$ that are generated from terms with different $X$ components, and solve for a satisfying $x$. This can be done efficiently, e.g., using Gaussian elimination. The absence of a consistent solution implies the absence of a curing Pauli group element.

As the dimension of each of the local terms $\{H_a\}$ is independent of the number of qubits $n$, and there are at most $\textrm{poly}(n)$ number of these terms~\cite{Cubitt:2016vl}, the entire procedure takes $\textrm{poly}(n)$ time.

\section{Single-qubit rotations are not enough to cure every Hamiltonian}
\label{SI2}

Consider, e.g., the three Hamiltonians
\beq
H_{ij}=Z_i Z_j + X_i X_j\ , \qquad (i,j) \in \{ (1,2), (2,3), (3,1)\}\ .
\eeq
The sum of any pair of these Hamiltonians can be cured by single-qubit unitaries (e.g., $H_2=H_{1,2}+H_{2,3}$ can be cured by applying $U=Z_2$). 

In contrast, the frustrated Hamiltonian $H=H_{12}+H_{23}+H_{13}$ cannot be cured using any combination of single-qubit rotations.
To see this, we first note that the partial trace of a stoquastic Hamiltonian is necessarily stoquastic. By partial trace over single qubits of $H$, we conclude that in order for $H$ to be stoquastic, all three $H_{ij}$'s must be stoquastic. To find all the solutions that convert each of these Hamiltonians to a stoquastic Hamiltonian, we expand
$R_i(\theta_i)\otimes R_j(\theta_j) H_{ij} R_i^{T}(\theta_i)\otimes R_j^{T}(\theta_j)$
and note that it has  $\pm \sin{2 (\theta_i-\theta_j)}$ and  $\cos{2 (\theta_i-\theta_j)}$ as off-diagonal elements. Demanding that the rotated Hamiltonians are all stoquastic [so that $\sin{2 (\theta_i-\theta_j)}=0$] forces $\cos{2 (\theta_i-\theta_j)}=-1$ $\forall (i,j) \in \{ (1,2), (2,3), (3,1)\}$. 
But this set of constraints does not have a feasible point. To see this note that 
\bea
\cos{2 (\theta_1-\theta_3)}=\cos{2 (\theta_1-\theta_2)} \cos{2 (\theta_1-\theta_3)} - \sin{2 (\theta_1-\theta_2)} \sin{2 (\theta_1-\theta_3)} =-1 \times -1 + 0\times 0=+1\ .
\eea
Therefore $H$ cannot be made stoquastic using $2\times 2$ rotational matrices. Because of the relation between rotation and orthogonal matrices (see Appendix~\ref{SI4-a} below), we conclude that $H$ cannot be made stoquastic using $2\times 2$ orthogonal matrices.

\section{Conjugation by a product of $X$ operators}
\label{SI3}

Here we show that conjugating a Hamiltonian by a tensor product of $X$'s or identity operators only shuffles the off-diagonal elements without changing their values. 
\begin{lemma}
	Let $U_X= X^{a_1} \otimes \dots \otimes X^{a_n}$, where $a_i\in \{0,1\}$. The set of off-diagonal elements of a general $2^n\times 2^n$ matrix $B$, $\mathrm{OFF}(B)=\{[B]_{ij}| i,j \in \{0,1\}^n, i\neq j \}$, is equal to the set of off-diagonal elements of $U_X B U_X$ for all possible $\{a_i\}$. 
\end{lemma}
\begin{proof}
	Let $a=(a_1,...,a_n)$. Similarly, let $i$ and $j$  represent $n$-bit strings. The elements of the matrix $U_X B U_X$ are:
\begin{align}
	\begin{split}
	\bra{i}U_X B U_X\ket{j}&=\bra{i_1,\dots,i_n}U_X B U_X\ket{j_1,\dots,j_n}\\
	&=\bra{i_1 \odot a_1 ,\dots,i_n\odot a_1} B \ket{j_1\odot a_1,\dots,j_n\odot a_1}=\bra{i\odot a}B \ket{j\odot a}
	\end{split}
\end{align}
where $\odot$ denotes the XOR operation. Clearly for any fixed $a$, we have  $i \neq j \iff i\odot a \neq j\odot a$ and therefore:
\begin{align}
	\begin{split}
\mathrm{OFF}(U_X B U_X)&=\{[B]_{i\odot a,j\odot a}| i,j \in \{0,1\}^n, i\neq j \}
=\{[B]_{i\odot a,j\odot a}| i\odot a,j\odot a \in \{0,1\}^n, i\odot a\neq j\odot a \} \\
&=\{[B]_{i'j'}| i',j' \in \{0,1\}^n, i'\neq j' \}=\mathrm{OFF}(B)
	\end{split}
\end{align}
\end{proof}
A similar argument proves that $U_X$  shuffles the diagonal elements: $\mathrm{DIA}(U_X B U_X)=\mathrm{DIA}(B)$.

Therefore, conjugating a Hamiltonian by $U_X$ does not change the set of inequalities  one needs to solve to make a Hamiltonian stoquastic. 
As a consequence, whenever we want to pick  $u_i$ as a solution, we can instead choose $Xu_i$.

\section{Completing the proof of Theorem~\ref{thm2}}
\label{SI4}

We provide the details needed to complete the proof of Theorem~\ref{thm2} as referenced in the main text. We first restate the construction and Theorem~\ref{thm2} for the reader's convenience.

The 6-local Hamiltonian $\tilde{\bar{H}}_{\mathrm{3SAT}}$  is defined as 
\bea
\tilde{\bar{H}}_{\mathrm{3SAT}}= \sum_{C} \bar{H}_{ijk}^{(\alpha\beta\gamma)}+ c  H'_{0} \ , \quad H'_{0}=  -\sum_{i=1}^{n} 2 \bar{Z}_i + \bar{X}_i\ ,
\label{eq:tildebarH3SAT}
\eea
with $c = O(1)$. Letting $\bar{Z}_i \equiv Z_{2i-1}Z_{2i}$, $\bar{X}_i \equiv X_{2i-1}X_{2i}$, and $\bar{W}_i^{{\alpha}}\equiv W_{2i-1}^{{\alpha}}  \otimes W_{2i}^{{\alpha}}$ we define 
\bea
\bar{H}_{ijk}^{(111)} =\bar{Z}_i \bar{Z}_j \bar{Z}_k &-& 3 (\bar{Z}_i +\bar{Z}_j +\bar{Z}_k)\nonumber \\
&-&(\bar{Z}_i \bar{Z}_j+ \bar{Z}_i \bar{Z}_k+\bar{Z}_j \bar{Z}_k) \ ,
\label{eq:restateH'ijk111}
\eea

and 

\beq
\label{eq:restatebarHijkabc}
\bar{H}_{ijk}^{(\alpha\beta\gamma)} =   
\bar{W}_i^{\bar{\alpha}} \otimes \bar{W}_j^{\bar{\beta}} \otimes \bar{W}_k^{\bar{\gamma}} \cdot \bar{H}_{ijk}^{(111)} \cdot \bar{W}_i^{\bar{\alpha}} \otimes \bar{W}_j^{\bar{\beta}} \otimes \bar{W}_k^{\bar{\gamma}}\ .
\eeq
$\bar{H}_{ijk}^{(\alpha\beta\gamma)}$ is a clause Hamiltonian corresponding to a clause in the 3SAT instance. Defining  $\bar{W}(x)\equiv\bigotimes_{i=1}^{n} \bar{W}_i^{x_i}$ it is  straightforward to check that $\bar{W}(x) \bar{H}_{ijk}^{(\alpha\beta\gamma)}\bar{W}(x)$ is stoquastic for any combination of the binary variables $(x_i,x_j,x_k)$ except for $(\alpha,\beta,\gamma)$.

In what follows we prove that deciding whether there exists a curing rotation $R=\bigotimes_{i=1}^{n} R(\theta_i)$ for $6$-local Hamiltonians is NP-complete. We do this by proving that: 

(i) Any satisfying assignment of a 3SAT instance provides a curing $R$ for the corresponding Hamiltonian $\tilde{\bar{H}}_{\mathrm{3SAT}}$, and 

(ii) Any $R$ that cures $\tilde{\bar{H}}_{\mathrm{3SAT}}$ provides a satisfying assignment for the corresponding 3SAT instance.

\textbf{Theorem 2 (restated).} Let $Q=\bigotimes_{i=1}^n q_i$, where each $q_i$ belongs to the single-qubit orthogonal group (i.e., $q_i^T q_i = I)$. Deciding whether there exists a curing orthogonal transformation $Q$ for $6$-local Hamiltonians is NP-complete.

\subsection{The relation between a single-qubit real-orthogonal matrix and a single qubit rotation}
\label{SI4-a}

The condition $q_{i} q^T_{i}=I$ forces each real-orthogonal matrix $q_i$ to be either a reflection or a rotation of the form:
\bea q_i=
\begin{bmatrix}
\cos{\theta_i} & a_i \sin{\theta_i}\\ 
 \sin{\theta_i} & -a_i \cos{\theta_i}
\end{bmatrix}
\eea
with $a_i =+1$ (a reflection) or $a_i =-1$ (a rotation). The operators $X$, $Z$, and Hadamard, are included in the family with $a_i=1$; $I$ and $iY=XZ$ are in the family with $a_i=-1$. Note that $\forall H, \forall \theta_i : \, q_i(\theta_i) H q_i^T(\theta_i)=q_i(\theta_i + \pi) H q_i^T(\theta_i + \pi)$. Therefore the angles that cure a Hamiltonian are periodic with a period of $\pi$. Hence, it suffices to consider the curing solutions only in one period: $\theta_i \in (\frac{-\pi}{2},\frac{+\pi}{2}]$.

Next, observe that a reflection by angle $\theta_i$ can be written as
\bea
 \begin{bmatrix}
\cos{\theta_i} &  \sin{\theta_i}\\ 
 \sin{\theta_i} & - \cos{\theta_i}
\end{bmatrix}\nonumber
=
X \begin{bmatrix} \cos{\frac{\pi}{2} - \theta_i} & - \sin{\frac{\pi}{2} - \theta_i}\\ 
 \sin{\frac{\pi}{2} - \theta_i} &  \cos{\frac{\pi}{2} - \theta_i} \end{bmatrix}
 =X  R(\frac{\pi}{2} - \theta_i)\ ,
 \eea
where 
\bea R(\theta_i)=
\begin{bmatrix}
\cos{\theta_i} & - \sin{\theta_i}\\ 
 \sin{\theta_i} &  \cos{\theta_i}
\end{bmatrix} \ .
\eea

As was discussed in Sec.~\ref{SI3}, if $X  R(\frac{\pi}{2} - \theta_i)$ is a curing operator so is $R(\frac{\pi}{2} - \theta_i)$. Therefore any curing $Q=\bigotimes_{i=1}^n q_i$ provides a curing $R=\bigotimes_{i=1}^n R(\theta_i)$. Hence, the NP-completeness of the decision problem for $R$ implies the NP-hardness of the decision problem for $Q$, which is the statement of Theorem~\ref{thm2}.

\subsection{If the 3SAT instance has a satisfying assignment $x^*$ then $P(x^*) \tilde{\bar{H}}_{\mathrm{3SAT}} P^T(x^*)$ is stoquastic}
\label{SI4-b}

Let 
\beq
P(x) \equiv\bigotimes_{i=1}^{n} \left(R(\frac{\pi}{4})^{x_i} \otimes  R(\frac{\pi}{4})^{x_i}\right)\ ,
\eeq 
where $x$ is a $n$-bit string $x=(x_1,\dots,x_n)$. Note that $R(\frac{\pi}{4})=XW$, and as we discussed in Appendix~\ref{SI3} it is equivalent to $W$ for curing.


To prove the claim, note that $x^*$ necessarily satisfies each individual clause of the 3SAT instance, and therefore makes the corresponding clause Hamiltonian stoquastic, i.e., $P(x^*) \bar{H}_{ijk}^{(\alpha\beta\gamma)} P^T(x^*)$  is stoquastic $\forall C_{ijk}^{(\alpha\beta\gamma)}$. Also, $P(x) H'_0 P^T(x)$ is clearly stoquastic for any $x$, where $H'_{0}=  -\sum_{i=1}^{n} 2 \bar{Z}_i + \bar{X}_i$ [Eq.~\eqref{eq:tildebarH3SAT}]. The stoquasticity of $P(x^*) \tilde{\bar{H}}_{\mathrm{3SAT}} P^T(x^*)$ then follows immediately.

\subsection{A useful lemma}
\label{SI4-c}
To prove the reverse direction, we first prove the following:

\begin{lemma} 
\label{lemmacont}

Let $\theta_1,\theta_2 \in (-\frac{\pi}{2},\frac{\pi}{2}]$. Consider the $2n$-qubit Hamiltonian
\bea
\label{eq:H'}
 H'=Z_1 Z_2 \otimes H_z + X_1 X_2 \otimes H_x + c H'_{12} \otimes I + I_1 I_2 \otimes H_I\ ,
\eea
where $H_z$, $H_x$, and $H_I$ are arbitrary Hamiltonians, 
 $c$ is an appropriately chosen constant, and 
\beq
H'_{12}= - 2 Z_1 Z_2 -  X_1 X_2\ .
\eeq
The only rotation $R(\theta_1) \otimes R(\theta_2)$ that can cure $H'$ has angles given by the following four points: 
\bea
(\theta_1,\theta_2) \in \left\{ (\frac{\pi}{2},\frac{\pi}{2}),(\frac{\pi}{4},\frac{\pi}{4}),(0,0), (\frac{-\pi}{4},\frac{-\pi}{4})\right\}.
\eea
\end{lemma}


\begin{proof}
It is straightforward to check that 	
\bes
\label{rotxz}
\begin{align}
&R(\theta_i) X R(-\theta_i)= \cos 2\theta_i X - \sin 2\theta_i Z,\\
&R(\theta_i) Z R(-\theta_i)= \sin 2\theta_i X + \cos 2\theta_i Z \ .
\end{align}
\ees
Using this we have:
\begin{align}
\label{eq:cases}
\begin{split}
	H'' &= [R(\theta_1) \otimes R(\theta_2)] H' [R(-\theta_1) \otimes R(-\theta_2)] \\
	&= X_1 X_2 \otimes [\sin 2 \theta_1 \sin 2 \theta_2  (H_z-2c I)+\cos 2 \theta_1 \cos 2 \theta_2  (H_x-cI)]\\
	&+X_1 Z_2 \otimes [\sin 2 \theta_1 \cos 2 \theta_2  (H_z-2c I)- \cos 2 \theta_1 \sin 2 \theta_2  (H_x-cI)]\\
	&+Z_1 X_2 \otimes [\cos 2 \theta_1 \sin 2 \theta_2  (H_z-2c I)- \sin 2 \theta_1 \cos 2 \theta_2  (H_x-cI)] \\
	&+Z_1 Z_2 \otimes [\cos 2 \theta_1 \cos 2 \theta_2  (H_z-2c I)+ \sin 2 \theta_1 \sin 2 \theta_2  (H_x-cI)]\\
	&+I_1 I_2 \otimes H_I\ .
	\end{split}
\end{align}
We next consider 
necessary conditions that $\theta_1$ and $\theta_2$ must satisfy in order to make 
$H''$ stoquastic. 

Let $A$ to $E$ be arbitrary matrices, and let $[0]$ denote the all-zero matrix. First, we note that the matrix $X_1 Z_2 \otimes B$ has both $+B$ and $-B$ as distinct off-diagonal elements for any non-zero matrix $B$ (a similar observation holds for $Z_1 X_2 \otimes C$). Second, we note that there are no common off-diagonal elements between $X_1 X_2 \otimes A$, $X_1 Z_2 \otimes B$, and $Z_1 X_2 \otimes C$. Third, there is no common off-diagonal elements between these three matrices and $Z_1 Z_2 \otimes D$ and $I_1 I_2 \otimes E$. Therefore, these terms cannot make the $\pm B$ in $X_1 Z_2 \otimes B$ non-positive. 

Thus, for $H''$ to become stoquastic it is necessary to have $B=[0]$:
\bea \label{eq:XZ}
\sin 2 \theta_1 \cos 2 \theta_2  (H_z-2c I)- \cos 2 \theta_1 \sin 2 \theta_2  (H_x-cI)=[0]\ .
\eea
Similar reasoning for $Z_1 X_2 \otimes C$ yields:
\bea \label{eq:ZX}
\cos 2 \theta_1 \sin 2 \theta_2  (H_z-2c I)- \sin 2 \theta_1 \cos 2 \theta_2  (H_x-cI)=[0]\ .
\eea

If the absolute value of at least one element of $H_z-2c I$ is different from the absolute value of the corresponding element in $H_x-cI$ (as discussed below, this is always possible by choosing an appropriate $c$), comparing the corresponding expressions from Eq.~\eqref{eq:XZ} and Eq.~\eqref{eq:ZX} we can conclude that: 
\bea \label{eq:sincos}
 \sin 2 \theta_1 \cos 2 \theta_2= \cos 2 \theta_1 \sin 2 \theta_2=0\ .
\eea
%
Equation~\eqref{eq:sincos}  gives rise to two possible cases: (i) $ \sin 2 \theta_1 = \sin 2 \theta_2=0$, or (ii) $\cos 2 \theta_1 = \cos2 \theta_2=0$. For $\theta_1,\theta_2 \in (-\frac{\pi}{2},\frac{\pi}{2}]$ these are the eight possible solutions:
\bea
(\theta_1,\theta_2) \in \{0,\frac{\pi}{2}\} \times \{0,\frac{\pi}{2}\} , \quad \text{or} \quad 
 (\theta_1,\theta_2) \in \{\pm\frac{\pi}{4}\} \times \{\pm\frac{\pi}{4}\} \ .
\eea
Now we observe that for an appropriately chosen $c$ (the value is discussed below), four of these solutions generate 
additional constraints on the allowed values of $\theta_1$ and $\theta_2$. 
To see this, we consider the $X_1 X_2 \otimes A$ term in Eq.~\eqref{eq:cases} in the two possible cases. 
For case (i), when $ \sin 2 \theta_1 = \sin 2 \theta_2=0$ and hence $\cos 2 \theta_1 =\pm 1$ and $\cos 2 \theta_2 =\pm 1$, this term becomes $X_1 X_2 \otimes \cos 2 \theta_1 \cos 2 \theta_2  (H_x-cI)]$. 
If $H_x-cI$ has any negative element, then the two combinations such that $\cos 2 \theta_1 \cos 2 \theta_2=-1$ flip the sign of this element and make the term non-stoquastic. I.e., if $H_x-cI$ has any negative element the only rotations that can keep $H''$ stoquastic satisfy $\cos 2 \theta_1 = \cos 2 \theta_2=1$ or  $\cos 2 \theta_1 = \cos 2 \theta_2=-1$.
Similarly, for case (ii), when $\cos 2 \theta_1 = \cos2 \theta_2=0$, if $H_z-2cI$ has any negative elements, the only rotations that can keep $H''$ stoquastic satisfy $\sin 2 \theta_1 = \sin 2 \theta_2=1$ or  $\sin 2 \theta_1 = \sin 2 \theta_2=-1$.

To summarize, for an appropriate choice of $c$, 
the solutions are necessarily one of these four points:
\bes
\begin{align}
\sin (2 \theta_1)=0, \cos(2 \theta_1)=1,\sin (2 \theta_2)=0, \cos(2 \theta_2)=1  \quad &\Longrightarrow (\theta_1,\theta_2)=(0,0) \\
\sin (2 \theta_1)=0, \cos(2 \theta_1)=-1,\sin (2 \theta_2)=0, \cos(2 \theta_2)=-1  \quad &\Longrightarrow (\theta_1,\theta_2)=(\frac{\pi}{2},\frac{\pi}{2})\\
\sin (2 \theta_1)=1, \cos(2 \theta_1)=0,\sin (2 \theta_2)=1, \cos(2 \theta_2)=0  \quad &\Longrightarrow (\theta_1,\theta_2)=(\frac{\pi}{4},\frac{\pi}{4})\\
\sin (2 \theta_1)=-1, \cos(2 \theta_1)=0,\sin (2 \theta_2)=-1, \cos(2 \theta_2)=0  \quad &\Longrightarrow (\theta_1,\theta_2)=(\frac{-\pi}{4},\frac{-\pi}{4})
\end{align}
\ees

The constraints on $c$ can be summarized as follows:
\begin{enumerate}
\item The absolute value of at least one element of $H_z-2c I$ is different from the absolute value of the corresponding element in $H_x-cI$, i.e., $(H_x-c I)^{\odot 2}\neq (H_z-2c I)^{\odot 2}$. (Here $\odot$ denotes the entrywise matrix multiplication.) 
\item Both  $H_x-cI$ and  $H_z-2cI$ have at least one negative element.
\end{enumerate}


Let $R'=\bigotimes_{i=3}^{2n} R(\theta_i)$ [we exclude $i=1,2$ because of the form of $H'$ in Eq.~\eqref{eq:H'}]. For our construction we have $H'=R' \tilde{\bar{H}}_{\mathrm{3SAT}} {R'}^{T}$. The term $\bar{Z}_1 \otimes H_z$ (recall that $\bar{Z}_i \equiv Z_{2i-1}Z_{2i}$ and $\bar{X}_i \equiv X_{2i-1}X_{2i}$) is composed of rotated 3SAT clause-Hamiltonians that share $\bar{x}_1$ in their corresponding 3SAT clauses. 
For example, we have $\bar{H}_{1jk}^{(111)}=\bar{Z}_1 \otimes (\bar{Z}_j \bar{Z}_k - 3 -  \bar{Z}_j - \bar{Z}_k) + I_1 I_2 \otimes H_I$ and therefore $R' (\bar{Z}_j \bar{Z}_k - 3 -  \bar{Z}_j - \bar{Z}_k) {R'}^{T}$ is included in $H_z$. Similarly, for a general clause that shares $\bar{x}_1$, namely $\bar{H}_{1jk}^{(1\beta\gamma)}$, $R' \bar{W}_j^{\bar{\beta}} \otimes \bar{W}_k^{\bar{\gamma}} (\bar{Z}_j \bar{Z}_k - 3 -  \bar{Z}_j - \bar{Z}_k) \bar{W}_j^{\bar{\beta}} \otimes \bar{W}_k^{\bar{\gamma}} {R'}^{T}$ would be included in $H_z$ (corresponds to a possible replacement of some of $\bar{Z}$ operators by $\bar{X}$ operators). 
It is straightforward to check that all the diagonal elements of each of these rotated clauses are non-positive. Namely, using Eq.~\eqref{rotxz} it is straightforward to check that the max norm, defined as $\|A\|_{\max}=\max_{ij} |[A]_{ij}|$, of any rotated Pauli operator is at most $1$, and therefore the same is true for any tensor product of rotated Pauli operators. In each 3SAT clause-Hamiltonian there are $3$ non-identity Pauli terms. Therefore they cannot generate a diagonal element that is larger than $3$. There is a $-3$ term for each clause, guaranteeing that all the diagonal terms remain non-positive.

Summing all these possible terms we have:
\bea
H_z=R' \left( \sum_{C_{1jk}^{(1\beta\gamma) } \in C} \bar{W}_j^{\bar{\beta}} \otimes \bar{W}_k^{\bar{\gamma}} (\bar{Z}_j \bar{Z}_k - 3 -  \bar{Z}_j - \bar{Z}_k) \bar{W}_j^{\bar{\beta}} \otimes \bar{W}_k^{\bar{\gamma}} \right) {R'}^T\ .
\eea
Based on the previous arguments, we conclude that all the diagonal elements of $H_z$ are non-positive. Furthermore, using the cyclic property of the trace and noting that except the $-3$  term all the other terms are traceless, we have $\Tr(H_z)= -3k$ with $k \in \mathbb{N}_0$ ($k=0$ only if $H_z=0$, i.e., when there is no $\bar{x}_1$ in any of the 3SAT-clauses).

$H_x$ is similar to $H_z$ but with a sum over $C_{1jk}^{(0\beta\gamma)}$. Using similar arguments, we conclude that all the diagonal elements of $H_x$ are all non-positive and $\Tr(H_x)= -3k'$ with $k' \in \mathbb{N}_0$.

Let us now show that $c=1$ satisfies the first of the two aforementioned constraints. As all the diagonal elements of $H_x$ and $H_z$ are non-positive, the first condition on $c$ is guaranteed if we find $c$ such that at least one of the diagonal element of $H_z-2cI$ and $H_x - cI$ are different. Note that  $\Tr(H_x-I)= -3k -2^{2n-2}$ and $\Tr(H_z-2I)= -3k' -2^{2n-1}$ where $k,k' \in \mathbb{N}_0$.  Obviously, these two traces cannot be equal for any value of $k$ and $k'$. 
 As the traces are different we conclude that at least one  diagonal element of $H_x-I$ is different form the corresponding element of $H_z-2I$. Therefore $c=1$ satisfies this constraint.

The second constraint is  satisfied for any $c>0$, and in particular $c=1$.
To see this note that all the diagonal elements of $H_x$ and $H_z$ are non-positive. Therefore we conclude that $H_x-cI$ and $H_z-2cI$ both have negative diagonal elements if  $c>0$.

\end{proof}

\subsection{Any $R$ that cures $\tilde{\bar{H}}_{\mathrm{3SAT}}$ provides a satisfying assignment for the corresponding 3SAT instance: restated}

Lemma~\ref{lemmacont} shows that for any $\tilde{\bar{H}}_{\mathrm{3SAT}}$, any curing rotation $R=\bigotimes_{i=1}^{2n} R(\theta_i)$ has to satisfy the condition that $(\theta_{2i-1},\theta_{2i}) \in \{ (\frac{\pi}{2},\frac{\pi}{2}),(\frac{\pi}{4},\frac{\pi}{4}),(0,0), (\frac{-\pi}{4},\frac{-\pi}{4})\}$ $ \forall i$.  If $(\theta_{2i-1},\theta_{2i})\in\{(0,0),(\frac{\pi}{2},\frac{\pi}{2})\}$ we assign $x_i=0$, while if $(\theta_{2i-1},\theta_{2i})\in\{(-\frac{\pi}{4},-\frac{\pi}{4}),(\frac{\pi}{4},\frac{\pi}{4})\}$ we assign $x_i=1$, since rotations with the angles in each pair have the same effect. 

As a consequence of Lemma~\ref{lemmacont} and the constraints it enforces on the possible solutions, the only way that a curing rotation can make a clause Hamiltonian non-stoquastic is if it would result in a $+\bar{X}_{i}\bar{X}_{j}\bar{X}_{k}$ term.  Since no other 3SAT-clause Hamiltonian in $\tilde{\bar{H}}_{\mathrm{3SAT}}$ contains an identical $\bar{X}_{i}\bar{X}_{j}\bar{X}_{k}$ term, these positive off-diagonal elements cannot be cancelled out or made negative regardless of the choice of the other variables in the assignment. Therefore if $R$ cures $\tilde{\bar{H}}_{\mathrm{3SAT}}$ it also necessarily separately cures all the Hamiltonian clauses in $\tilde{\bar{H}}_{\mathrm{3SAT}}$. By construction if $R$ cures a Hamiltonian clause $\bar{H}_{ijk}^{(\alpha\beta\gamma)}$, the string $x$ constructed above, satisfies the corresponding 3SAT clause $C_{ijk}^{(\alpha\beta\gamma)}$. Thus $x$ satisfies  all the clauses in the corresponding 3SAT instance.


\section{A potential encryption protocol based on `secretly stoquastic' Hamiltonian}
\label{SI5}

By generating 3SAT instances with  planted solutions (see, e.g., Refs.~\cite{Hen:2015rt,Hamze:2017aa}) and transforming these to non-stoquastic Hamiltonians via the mapping prescribed by Theorems~\ref{thm1} (or ~\ref{thm2}), one would be able to generate $3$-local (or $6$-local) Hamiltonians that are stoquastic but are computationally hard to transform into stoquastic form. 

This construction may have cryptographic implications. For example, imagine planting a secret $n$-bit message in the (by design unique) ground state of a stoquastic Hamiltonian. Since the solution is planted, Alice automatically knows it. She checks that QMC can find the ground state in a prescribed amount of time $\tau(n)$, and if this is not the case, she generates a new, random, stoquastic Hamiltonian with the same planted solution and checks again, etc., until this condition is met. Alice and Bob pre-share the secret key, i.e., the curing transformation, and after they separate Alice transmits only the $O(n^2)$ coefficients of the non-stoquastic Hamiltonians (transformed via the mapping prescribed by Theorem~\ref{thm1}) for every new message she wishes to send to Bob. To discover Alice's secret message, Bob runs QMC on the cured Hamiltonian. Since Alice verified that QMC can find the ground state in polynomial time, Bob will also find the ground state in polynomial time. 

This scheme should be viewed as merely suggestive of a cryptographic protocol, since as it stands it contains several potential loopholes: (i) Its security depends on the absence of efficient two-or-more qubit curing transformations, as well as the absence of algorithms other than QMC that can efficiently find the ground state of the non-stoquastic Hamiltonians generated by Alice; (ii) The fact that Alice must start from Hamiltonians for which the ground state can be found in polynomial time may make the curing problem easy as well; (iii) this scheme transmits an $n$-bit message using $r n^2$ message bits, where $r$ is the number of bits required to specify the $n^2$ coefficients of the transmitted non-stoquastic Hamiltonian, so it less efficient than a one-time pad. Additional research is needed to improve this into a scheme that overcomes these objections.

\end{document}